\documentclass[11pt]{amsart}
\usepackage{geometry}                % See geometry.pdf to learn the layout options. There are lots.
\usepackage{graphicx}

\usepackage{amsmath,amsthm,amssymb}
\usepackage{epstopdf}
\newtheorem{lemma}{Lemma}
\newtheorem{theorem}{Theorem}
\def \D {Y_{\eta} \times \mathbb{R}}

\title{Rigorous asymptotic study of the screened electrostatic potential in a thin dielectric slab}
\author{Didier Felbacq, Emmanuel Rousseau \\ L2C, Univ Montpellier, France}
\def \e{\varepsilon}

\def \bro{\boldsymbol\rho}
\def \f{\hat{\varphi}}
\def \gf{\varphi}
\def \k{\mathbf{k}}
\def \fl {\rightarrow}
\def \be{\begin{equation}}
\def \ee{\end{equation}}
\begin{document}
\begin{abstract}
The screened Coulomb potential plays a crucial role in the binding energies of excitons in a thin dielectric slab. The asymptotic behavior of this potential is studied when the thickness of the slab is very small as compared to the exciton Bohr radius. A regularized expression is given and the exact effective 2D potential is derived. These expressions may be useful for the computation of the exciton binding energy in 2D or quasi-2D materials.
\end{abstract}
\maketitle
\section{Introduction}
The raise of 2D materials has generated a renewed interest in computing excitons binding energy in thin structures, such as a slab of dielectric material. In that situation, it has been well appreciated in the literature that the Coulomb potential is screened and that it plays a crucial role on the value of the binding energies \cite{thygesen2017,thygesen2015,rubio2010,rubio2011}. Indeed, within the effective mass approximation, the enveloppe of the exciton wave function satisfies a Schrödinger equation involving the screened Coulomb potential. Several approaches have been used to derive an expression for this potential. One of the oldest work on this is an article by Keldysh \cite{Keldysh}. In order to apply this potential to the situation of 2D materials, several approaching have been put forward in order to obtain a strictly 2D potential. In \cite{Keldysh}, the space variables transverse to the material are abruptly put to 0. In \cite{rubio2010,rubio2011} the starting point is a 2D polarisability, whereas in \cite{thygesen2015} the charge distributions are described as lines of charge, in order to obtain an effective 2D potential. Therefore, it seems that a direct limit analysis of the 3D potential when the width of the slab is very small is lacking.

 In this note, we propose a multiple scale approach to the study of the screened potential when the thickness of the slab is very small with respect to the Bohr radius. We obtain a regularized potential by exhibiting explicitly the singular part, i.e. the bare Coulomb potential. The 2D effective potential compares very well with the approximation given in \cite{Keldysh}. 
\section{Expression of the screened potential and regularization}
\subsection{The screened potential}

Consider the geometry described in fig.\ref{chema}. Cylindrical coordinates $(\bro,z)$ are used. We consider the electrostatic interaction between two charges $e$ located respectively at $(\bro_0,z_0)$ and $(0,z'_0)$. We assume that $z_0 \geq z_0'$. The particles are situated in a dielectric slab of permittivity $\e_f$ and width $d$ situated in the interval $z \in [-d/2,d/2]$. The slab is surrounded by two semi-infinite medium of permittivities $\e_1$ for $z<-d/2$ and $\e_2$ for $z>d/2$. 
We denote $$\eta_1=\frac{1}{2} \log \left(\frac{\e_f+\e_1}{\e_f-\e_1}\right), \, \eta_2=\frac{1}{2} \log \left(\frac{\e_f+\e_2}{\e_f-\e_2}\right),$$
and
$$
Y_{\eta}=[-\eta/2,\eta/2]^2  \hbox{ and } Y=[-1/2,1/2]^2.
$$
The electrostatic potential energy between the particles is defined by: $V(\bro,z_0,z'_0)= e' \f_{e}(0,z'_0)=e \f_{e'}(\bro,z_0)$, where $\f_{q}$ is the potentiel created by particle $q$. 

\begin{figure}[ht] 
\begin{center} 
\includegraphics[width=8cm]{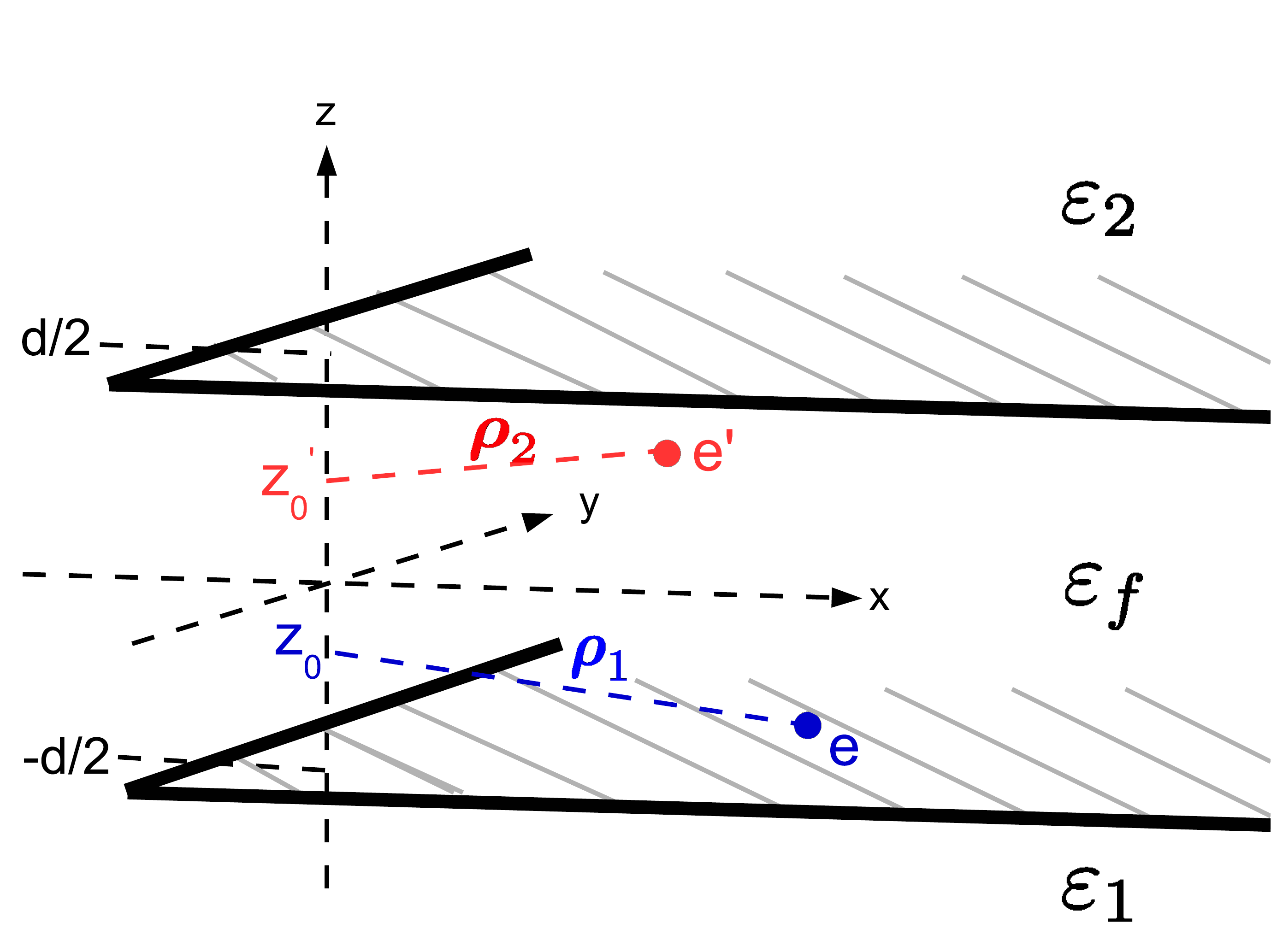}
\caption{Scheme of the structure under study. We consider the electrostatic interaction between two electric charges $e$ and $e'$ located respectively at position $(\bro_1,z_0)$ and $(\bro_2,z_0')$ in a slab of dielectric constant $\varepsilon_f$. The slab width is $d$. It is surrounding by a medium of dielectric constant $\varepsilon_1$ ($z \le -d/2$) and a medium of dielectric constant $\varepsilon_2$ ($z \ge -d/2$).}\label{chema}
\end{center} 
\end{figure}

In the appendix, we give a full derivation of the potential since in \cite{Keldysh} the result is stated without details. After all calculations are performed (see the appendix) the final result is

\begin{equation}
 V(\rho,z_0,z'_0)=\frac{e\,e'}{2\pi \, \e_f d} \,   I(\frac{\rho}{d},\frac{z_0}{d},\frac{z'_0}{d}),
 \end{equation}
 where 
 \be
I(r,x,y)=\int_0^{+\infty} W(u,x,y) J_0(r u) du,
\ee
and the kernel $W$ is given by
$$
W(u,x,y)=\frac{\cosh\left[u(\frac{1}{2}+y) +\eta_1\right]\,\cosh\left[u(\frac{1}{2}-x)+\eta_2\right]}{\sinh\left(u+\eta_1+\eta_2\right)},\, (x,y) \in Y.
$$
 Since the expression was obtained for $z_0>z'_0$ (see appendix), these variables should be switched when $z'_0>z_0$.

From now on, we assume $e=-e'$, that is, the second particle is a hole.
 
 \subsection{Solving the Schr\"odinger equation-Defining new coordinates}
In a generic case, the situation considered in fig.\ref{chema} is a two-body problem. The particles have masses $m_1$ and $m_2$ and are parametrized in cylindrical coordinates by $(\bro_0,z_0)$ and $( \bro'_0,z'_0)$ . We look for a time-independent wave function $\psi(\bro_0,z,\bro'_0,z')$ satisfying:
$$
-\frac{\hbar^2}{2m_1}\Delta_0 \psi-\frac{\hbar^2}{2m_2}\Delta'_0 \psi +V(|\bro_0-\bro'_0|,z_0,z'_0) \psi=E\psi
$$
As usual in this kind of problem, we focus on the relative motion. Here however, we cannot do this for the height variables $z$ since the potential does depend separately on $z_0$ and $z'_0$. Therefore, we denote $\bro=\bro_0-\bro'_0$ and we obtain after eliminating the movement of the center of gravity in the $xOy$ plane:
$$
 \left\{ -\Delta_{||} - \kappa_1 \partial^2_{z_0}- \kappa_2 \partial^2_{z'_0}+ \frac{2\mu}{\hbar^2}V(\rho,z_0,z'_0)\right\} \psi=\frac{2\mu}{\hbar^2} E \psi
$$
where $\mu = m_1 m_2/(m_1+m_2)$ is the reduced mass, $\kappa_1=\mu/m_1,\, \kappa_2=\mu/m_2$ and $\rho=|\bro|$. 

Let us now normalize the variables relatively to the slab width $d$. We define the new variables $z_0=\tilde{z} d,\, z'_0=\tilde{z}' d$. These variables belong to the interval $[-1/2,1/2]$. We also normalize the in-plane variable and define: $\rho=r d$. Plugging these into the equation leads to
$$
 \left\{  -\Delta_{||} - \kappa_1 \partial^2_{\tilde{z}}- \kappa_2 \partial^2_{\tilde{z}'}- \frac{2 \mu  d^2}{\hbar^2} V(r,z,z')\right\} =\frac{2\mu d^2}{\hbar^2} E \psi.
$$
These basic transformations allow to introduce the exciton Bohr radius \be a_0=4 \pi \hbar^2 \epsilon_f/\mu e^2,\ee quite naturally. 

From now on, we denote $\eta=d/a_0$ which is the small parameter of our problem. 
The final variables are now $\zeta$ and $\zeta'$, satisfying $\zeta=\tilde{z} \eta,\, \zeta'=\tilde{z}' \eta$. These new variables belong to $[-\eta/2,\eta/2]$. 

The final spectral equation with a small parameter is
\be
 \left\{  -\Delta_{||} - \eta^2 \kappa_1  \partial^2_\zeta- \eta^2 \kappa_2 \partial^2_{\zeta'}- 4 \eta I(r,\zeta/\eta,\zeta'/\eta) \right\} \psi=E_{\eta} \psi,
\ee
where: $E_{\eta}=(2\mu d^2/\hbar^2)  E$. 
\subsection{Regularization of the potential}
It is clear that when the slab has an infinite width, one should recover the usual Coulomb potential $V_C(r)=1/4\pi \e_f r$. Therefore we expect that $V$ be a perturbation of $V_C$. In this section, we exhibit the singular part of $V$.
First, we establish two lemma.
\begin{lemma}
The kernel of screened Coulomb potential has the following asymptotic behavior
$$
W(u,x,y) \sim \frac{1}{2} e^{-u |x-y|} \hbox{ as } u \fl +\infty.
$$
\end{lemma}
\begin{proof}
This is quite elementary using the expression of $\sinh$ and $\cosh$ in terms of the exponential.
\end{proof}
In order to evaluate the singular part, we need the following result
\begin{lemma}\label{lembessel}
The following result holds for $\Re \alpha > |\Im \beta|$
$$
\int_0^{+\infty} e^{-\alpha x}\, J_0(\beta x) dx= \frac{\beta^{-\nu} [\sqrt{\alpha^2+\beta^2}-\alpha]^{\nu}}{\sqrt{\alpha^2+\beta^2}}.
$$
\end{lemma}
\begin{proof}
See \cite[p.702]{zwillinger} formula (6.23 (3)).
\end{proof}
We are now in a position to exhibit a regularized potential, that is, which is not singular at the origin
\begin{theorem}
The following decomposition holds
\be \label{CoulCorr}
I(r,x,y)=\int_0^{+\infty} \left(W(u,x,y)-\frac{1}{2} e^{-u |x-y| }\right) J_0(r u) du+\frac{1}{2 \sqrt{r^2+|x-y|^2}}.
\ee
The kernel of the first integral tends exponentially fast towards 0 and it defines a function that is regular near the origin $x=y=0$.
\end{theorem}
\begin{proof}
We substract the asymptotic behavior
$$
I(r,x,y)=\int_0^{+\infty} \left(W(u,x,y)-\frac{1}{2}e^{-u |x-y| }\right) J_0(r u) du+\frac{1}{2} \int_0^{+\infty} J_0(r u)e^{-u |x-y| } du
$$
Using lemma (\ref{lembessel}) we get
$$
 \int_0^{+\infty} e^{-u |x-y| }  J_0(r u)du=\frac{1}{\sqrt{r^2+|x-y|^2}}
$$
and the result follows.
\end{proof}

We have exhibited the screened electrostatic in a dielectric slab as the usual Coulomb potential plus a correcting term. This  term is exponentially small as the slab width is large compare to the relative distance between the two electric charges.
Furthermore our expression (\ref{CoulCorr}) is regularized and do not present any divergence as the height $z$ approaches zero. This expression is suitable for further numerical calculations aiming to compute the binding energy of an exciton in a 2D materials. In the following we show how a multiscale approach allows to obtain the 2D potential when the width of the slab is very small with respect to the exciton Bohr radius.

 \section{Asymptotics of the spectral problem}
 In the previous analysis two scales enter into the problem, namely the slab width $d$ and the height of the electric charges $z_0,z'_0$. Working with quantities normalized by the slab width $d$ helped us exhibiting the regular part of the 2D electrostatic potential. Finding the exciton binding energy introduces another scale: the exciton Bohr radius $a_0$. We are interested in problems where the Bohr radius is large as compared to the slab width: $a_0 \gg d$, that is, we focus on two-dimensional problems. Since now two length scales contribute to the problem, we need to introduce a new parameter $\eta=d/a_0$ in order to vary one length-scale independently of the other. Letting $\eta$ approach zero allows to deal with a 2D problem, by considering the electrostatic problem for a 2D sheet.
   
 \subsection{The multiple scale approach}
  We are now in a position to obtain the limit behavior of the Hamiltonian when $\eta$ tends to $0$.
  \begin{theorem}
  As $\eta \rightarrow 0$, the asymptotic 2D expansionn to first order in $\eta$, of the Hamiltonian describing the exciton is given by
  \be
  -\Delta_{||} - 4 \eta I_0(r),
  \ee
  where the effective 2D potential $I_0$ is given by
  \be \label{Vfinal}
  I_0(r)=\int_{Y} I(r,x,y) dx dy.
  \ee
  \end{theorem}
  \begin{proof}
  Consider the spectral problem
  $$
  \left\{-\Delta_{||} -\eta^2 \mu/m_1 \partial^2_\zeta-\eta^2 \mu/m_2 \partial^2_{\zeta'}+\eta V(r,\zeta/\eta,\zeta'/\eta)\right\}\psi_{\eta}=E_{\eta} \psi_{\eta}.
$$
In order to obtain the asymptotic behavior when $\eta \rightarrow 0$, we put this expression in variational form. To do so we use a test function $\phi(r,\zeta, \zeta')$ such that, for each $r \in [0,+\infty[$, the function $(\zeta,\zeta') \rightarrow \phi(r,\zeta,\zeta')$ belongs to $D(Y)$, i.e. the Schwartz space of $C^{\infty}$ functions with compact support in $Y$. 

It holds
\begin{eqnarray*}
\int_{\D} \nabla_{||}\psi_{\eta} \nabla \phi +\eta^2 \kappa_1 \int_{\D}  \partial_\zeta \psi_{\eta} \partial_\zeta \phi+\\
\eta^2 \kappa_2 \int_{\D}  \partial_{\zeta'} \psi_{\eta} \partial_\zeta \phi+\eta \int_{\D}  V(r,\zeta/\eta,\zeta'/\eta) \psi \phi=E_{\eta}  \int_{\D} \psi_{\eta} \phi.
\end{eqnarray*}
Let us now divide this equality by $\eta^2$. Using Lebesgue theorem, we know that, for a continuous summable function $f(x,y)$ defined on $Y_{\eta}$, one has: $\lim_{\eta \rightarrow 0} \eta^{-2} \int_{Y_{\eta}} f(x,y)=f(0,0)$. Therefore, we obtain
$$
\int_{Y_{\eta}}  \partial_\zeta \psi_{\eta} \partial_\zeta \phi=O(\eta^2),\, \int_{Y_{\eta}}  \partial_{\zeta'} \psi_{\eta} \partial_{\zeta'} \phi=O(\eta^2),
$$
and 
$$
\frac{1}{\eta^2}\int_{\D} V(r,\zeta/\eta,\zeta'/\eta) \psi_{\eta} \phi d\zeta d\zeta'=\int dr \left[ \int_{Y}  V(r,x,y)dxdy \right] \psi^0_{\eta}(r) \phi^0(r)+o(1),
$$
where
$$
\psi^0_{\eta}(r)= \psi_{\eta}(r,0,0),\,\phi^0(r)=\phi(r,0,0).
$$
We conclude that, up to order $\eta$, the variational relation is
$$
\int \nabla_{||}\psi^0_{\eta} \nabla_{||} \phi^0 +\eta \int dr \left[ \int_{Y}  V(r,x,y)dxdy \right] \psi^0_{\eta}(r) \phi^0(r)=E_{\eta}  \int \psi^0_{\eta} \phi^0.
$$
The result follows.
\end{proof}
\begin{figure}[ht] 
\begin{center} 
\includegraphics[width=10cm]{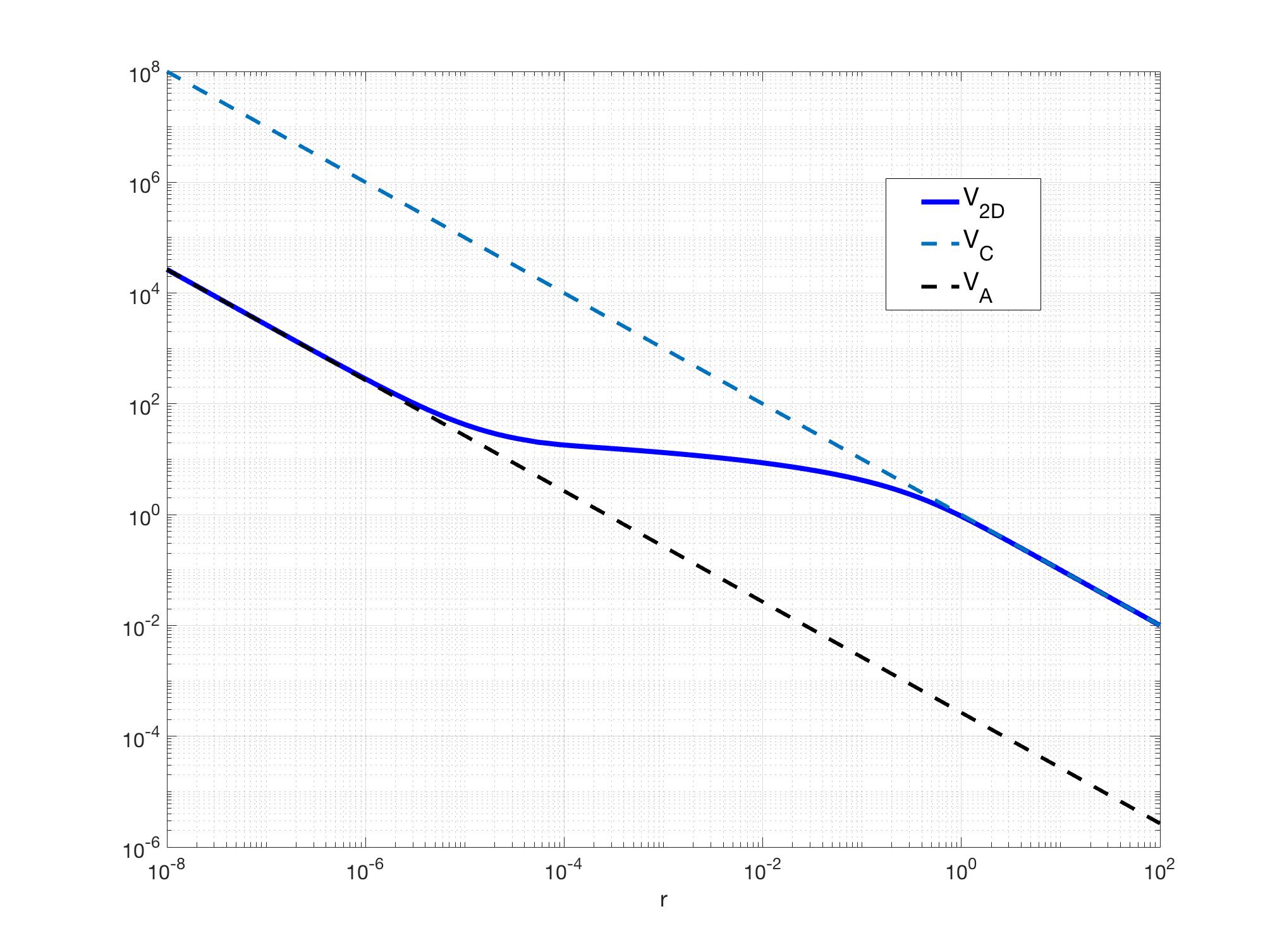}
\caption{LogLog plot of the effective 2D potential (\ref{Veff}). It is compared to the Coulomb potential and to the fitted potential (\ref{Vfit}).}\label{V2D}
\end{center} 
\end{figure}

 This result shows that contrarily to what could be intuitively believed, the effective 2D potential is not obtained from the 3D one simply by putting $z=z'=0$. Let us consider more specifically what happens with the strictly Coulomb part of the potential.
The effective potential in that case is
\be \label{Veff}
V_{2D}(r)= \int_{Y} \frac{1}{\sqrt{r^2+|z'-z|^2}}  dz dz' 
\ee
whereas by putting $z=z'=0$ one obtains simply $V_c=1/r$. In fig. (\ref{Veff}), we have plotted both $V_{2D}$ and $V_c$. There it can be seen that, as $r$ tends to $+\infty$, the effective potential behaves as the Coulomb potential. However, as $r$ tends to $0$ it behaves as
\be  \label{Vfit}
V_{A}= \frac{\lambda}{r},
\ee
where $\lambda=\frac{8}{3}\, 10^{-4}$, an expression that was obtained by fitting.

In \cite{Keldysh}, the following 2D approximation is introduced (the expression is adapted in order to take into account the change of unit system)
% The introduction of Struve and Neumann functions is totally ridiculous
%for $\rho\gg d$ denote $\delta=d/\rho$
$$ 
V_K(\rho)=\frac{ee'}{2 \pi \epsilon d} I_K(\rho),
$$
where 
\be\label{VK}
I_K(\rho)=\int_0^{+\infty} \frac{J_0(t) dt}{t+\frac{\epsilon_1+\epsilon_2}{\epsilon_f} r}.
\ee
 \begin{figure}[ht] 
\begin{center} 
\includegraphics[width=10cm]{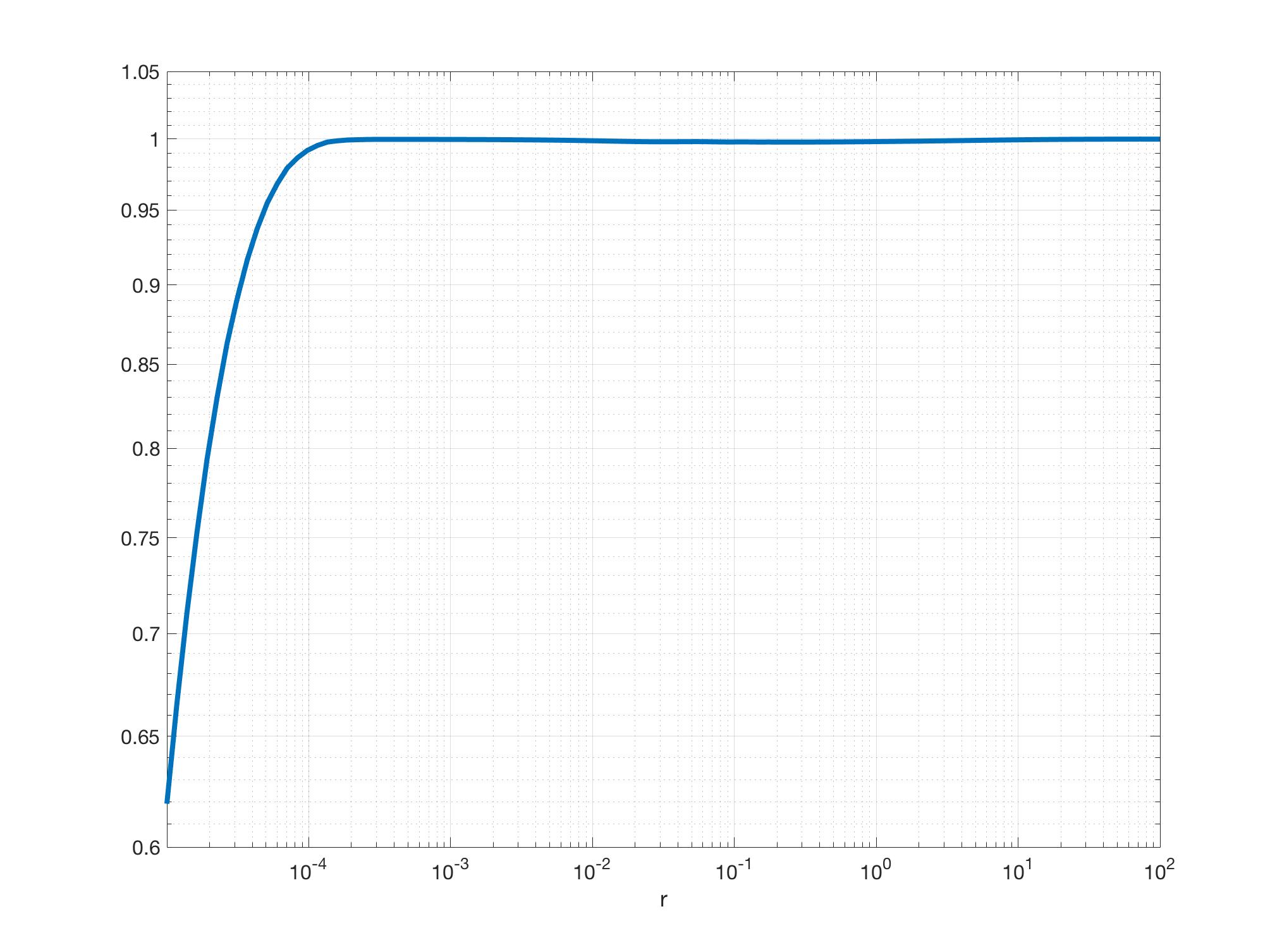}
\caption{LogLog plot of the ratio $I_K/I_0$ between the Keldysh potential (\ref{VK}) and the effective potential (\ref{Vfinal}).}\label{VmeanVK}
\end{center} 
\end{figure}

It can be expressed using Struve and Neumann functions \cite{Keldysh,pulci2015} but this representation is of little interest since the integral in (\ref{VK}) can be computed numerically very easily. This form of the 2D potential was used in several articles, see e.g. \cite{pulci2015}, in order to compute the binding energy of excitons. Several authors have derived this expression using various approaches \cite{thygesen2015,rubio2011}. Let us compare our result (\ref{Vfinal}) and this potential. Using $\epsilon_1=\epsilon_2=1$ and $\epsilon_f=10$. The ratio $I_K/I_0$ is plotted in fig. \ref{VmeanVK}. An excellent agreement is seen on a large range of values. It is only near the singularity that the potentials differ largely from each other.
Our expression for an effective 2D potential is exact, as it comes from a rigorous limit analysis. Its validity is not limited to specific values of the variable $r$. As such it can take into account precisely the interaction between particles at small distances.
 \section{Conclusion}
 We have derived the asymptotic behavior of the Hamiltonian for the exciton wave function in a thin dielectric slab pf width $d$. Our approach is based on a limit analysis of the Hamiltonian using a suitable small parameter $\eta=d/a_0$, where $a_0$ is the Bohr radius. Our result compares well with the known results of the literature, although it is not limited to asymptotic values of the permittivity contrast. The method used here could be extended to obtain an expansion of the Hamiltonian with respect to $\eta$ \cite{moi}. We have also provided a regularized expression for the full 3D potential which makes explicit to what extend this potential departs from the Coulomb ones. This could be especially useful for implementing a perturbation analysis. 
\section{Appendix: Derivation of the screened potential.}
\subsection{Expression of the Green function}
\subsubsection{The equation in the Fourier domain}
Let us compute then the potential $\gf(\bro,z;z'_0)$ created by a unit charge situated at $(0,z'_0)$. It satisfies the Poisson equation in the Schwartz distributions meaning:
\begin{equation}
-\nabla\cdot(\e(z)\nabla\gf)=\delta(z-z'_0)\otimes \delta(\bro)
\end{equation}
and the conditions at infinity: $\lim_{z \fl \pm \infty} \f=0$.

By invariance of the medium in the $\bro$ directions, a partial Fourier transform is performed. The Fourier transform is defined by:
$$
\f(\k,z;z'_0)=\int \gf(\bro,z;z'_0) e^{-i \k \cdot \bro} d^2\bro, 
$$
and the inverse transform is:
$$
\gf(\bro,z;z'_0)=\frac{1}{(2\pi)^2} \int \f(\k,z;z'_0) e^{i \k \cdot \bro} d^2\k
$$

The Fourier transform of $\f$ satisfies, in the distributional meaning, the following differential equation:
\begin{equation}\label{poisk}
-\partial_x (\e\, \partial_x \f)+k^2 \e \f=\delta(z-z'_0),
\end{equation}
where: $k=|\k|$. For simplicity, the dependence of $\f$ with respect to $\k$ and $z'_0$ is implicit: we denote $\f(z)$ instead of $\f(\k,z;z'_0)$.
\subsection{The boundary conditions}
There are four regions to be considered:
\begin{enumerate}
\item $z>d/2$, $-\f''+k^2 \f=0$, hence, taking into account the condition at infinity: $\f(z)=A_2 \,e^{-k(z-d/2)}$,
\item $z'_0<z<d/2$, $\f(z)=A^+_f \,e^{-k(z-d/2)}+B^+_f \,e^{k(z-d/2)}$,
\item $-d/2<z<z'_0$, $\f(z)=A^-_f \,e^{-k(z-d/2)}+B^-_f \,e^{k(z-d/2)}$,
\item $z<-d/2$, $\f(z)=B_1 \,e^{k(z+d/2)}$
\end{enumerate}
The boundary conditions at the interfaces of each domain are implied by eq.(\ref{poisk}):
the function $\f$ is continuous everywhere and the function $\e \partial_x \f$ is continuous everywhere except at $z'_0$ where it as a jump: $\partial_x \f(z_0^+)-\partial_x \f(z_0^-)=-1/\e$.

For $\f$ this gives the relations:
\begin{itemize}
\item  $z=d/2$, $\f(\left.\frac{d}{2}\right|^+)=\f(\left.\frac{d}{2}\right|^-)$: 
\begin{equation}\label{d/2} A_2=A_f^++B_f^+,\end{equation}
\item $z=z'_0$, $\f(\left.z'_0\right| ^+)=\f(\left.z'_0\right| ^-)$: 
\begin{equation}\label{z'_0} A_f^+\, e^{-k(z'_0+d/2)}+B_f^+\,e^{k(z'_0+d/2)}=
A_f^-\, e^{-k(z'_0-d/2)}+B_f^-\,e^{k(z'_0-d/2)}\, \end{equation}
\item $z=-d/2$, $\f(\frac{d}{2}^+)=\f(\frac{d}{2}^-)$: 
\begin{equation}\label{-d/2}
Af^-+B_f^-=B_1
\end{equation}

For the derivative, we obtain:
\item $z=d/2$,  \begin{equation}\label{dd/2}-\e_2 A_2=-\e_f A_f^++\e_f B_f^+,\end{equation}
\item $z=z'_0$, \begin{equation}\label{dz'_0}-k A_f^-\, e^{-k(z'_0+d/2)}+k B_f^-\,e^{k(z'_0+d/2)}-
(-k A_f^+ e^{-k(z_0-d/2)}+k B_f^+ e^{k(z_0-d/2)})=1/\e_f ,\end{equation}
\item $z=-d/2$, \begin{equation}\label{-dd/2}-\e_f Af^-+\e_f B_f^-=\e_1 B_1\end{equation}
\end{itemize}
From [(\ref{d/2}),(\ref{dd/2})] and  [(\ref{-d/2}),(\ref{-dd/2})],we get:
\begin{equation}\label{prop}
A_f^+=\frac{\e_f+\e_2}{\e_f-\e_2} \, B_f^+ ,\, A_f^-=\frac{\e_f-\e_1}{\e_f+\e_1} \, B_f^-
\end{equation}
In order to clarify the derivation, we denote:
\begin{equation*}
\tau_1=\frac{\e_f+\e_1}{\e_f-\e_1}, \, \tau_2=\frac{\e_f+\e_2}{\e_f-\e_2}
\end{equation*}
From (\ref{d/2}) and (\ref{prop}), we get:
\begin{equation}\label{B+-}
B^-_f =\tau B^+_f
\end{equation}
where 
\begin{equation}\label{tau}
\tau=\frac{\tau_2 \, e^{-k(z'_0-d/2)}+e^{k(z'_0-d/2)}}{\tau_1^{-1} e^{-k(z'_0+d/2)}+e^{k(z'_0+d/2)}}
\end{equation}
This last expression can be simplified. Let us denote: 
\begin{equation*}
\eta_{n}=\log(\sqrt{\tau_{n}}), \,\hbox{i.e. } \sqrt{\tau_{n}}=e^{\eta_{n}}, \, n=1,2.
\end{equation*}
From (\ref{tau}), we get:
\begin{equation}\label{taufinal}
\tau=\sqrt{\tau_1\,\tau_2} \, \frac{\cosh\left[k(d/2-z'_0) +\eta_2\right]}{\cosh\left[k(d/2+z'_0) +\eta_1\right]}
\end{equation}
We are looking first for the expression of $B_f^+$. It is obtained from (\ref{dz'_0}), by using [(\ref{prop})(\ref{B+-})(\ref{taufinal})]:
$$
k \tau B_f^+ \left[-\tau_1^{-1} e^{-k(z'_0+d/2)}+e^{k(z'_0+d/2)}\right]-k B_f^+ \left[-\tau_2 e^{-k(z'_0-d/2)}+e^{k(z_0-d/2)} \right]=\frac{1}{\e_f} 
$$
this gives, upon using the expression (\ref{taufinal}) for $\tau$
 \begin{equation*}
B_f^+=\frac{1}{2k\e_f\sqrt{\tau_2}}  \Gamma, A_f^+=\tau_2\, B_f^+\,,
 \end{equation*}
 where:
 \begin{equation}\label{gamma}
 \Gamma=\frac{\cosh\left[k(d/2+z'_0) +\eta_1\right]}{\sinh\left[k\,d+\eta_1+\eta_2\right]}.
 \end{equation}
 \subsection{Expression of the energy of interaction}
 The electrostatic energy between both charges is given by:
 \begin{equation*}
 V(\bro,z_0,z'_0)=\frac{e\,e'}{(2\pi)^2 \e_0} \int \f(\k,z_0;z'_0) e^{i \k\cdot \bro} d^2\k,
 \end{equation*}
For $z_0 \geq z'_0$, it holds:
$$
\f(\k,z_0;z'_0)=A^+_f \,e^{-k(z_0-d/2)}+B^+_f \,e^{k(z_0-d/2)}\, ,
$$
this gives:
\begin{eqnarray*}
\f(\k,z_0;z'_0)=\frac{\Gamma}{k\e_f} \cosh[k(d/2-z_0)+\eta_2]
\end{eqnarray*}
Using (\ref{gamma}), it comes:
$$
\f(\k,z_0;z'_0)=\frac{1}{k\e_f} \frac{\cosh\left[k(d/2+z'_0) +\eta_1\right]\,\cosh[k(d/2-z_0)+\eta_2]}{\sinh\left[k\,d+\eta_1+\eta_2\right]} \,
$$
and finally:
$$
 V(\bro,z_0,z'_0)=\frac{e\,e'}{4\pi^2 \, \e_f} \int \frac{\cosh\left[k(\frac{d}{2}+z'_0) +\eta_1\right]\,\cosh[k(\frac{d}{2}-z_0)+\eta_2]}{k\, \sinh\left(k\,d+\eta_1+\eta_2\right)} e^{i \k\cdot \bro} d^2\k,
$$
There is a typo in the expression given in the paper by Keldysh: in the integral defining the energy of interaction, the term $e^{2 \k \cdot \bro}$ should be replaced by $e^{i \k\cdot \bro}$.

Consider the double integral in polar coordinates: $(\rho,\theta)$: $\bro=\rho(\cos\theta,\sin\theta)$ and $k(\cos\psi,\sin\psi)$: $\bro \cdot \k=\rho\, k \cos(\theta-\psi)$. It is possible without loss of generality to take $\theta=\pi/2$, from invariance of the problem under a rotation around axis $Oz$.
\begin{equation*}
 V(\bro,z_0,z'_0)=\frac{e\,e'}{4\pi^2 \, \e_f} \int k dk \frac{\cosh\left[k(\frac{d}{2}+z'_0) +\eta_1\right]\,\cosh[k(\frac{d}{2}-z_0)+\eta_2]}{k\, \sinh\left(k\,d+\eta_1+\eta_2\right)} \int_0^{2\pi} d\psi \,e^{i k\rho \sin(\psi)} ,
 \end{equation*}
 Let us recall the generating series for Bessel functions $J_n(\rho)$:
 $$
 e^{i \rho \sin \psi}=\sum_n J_n(\rho) e^{in\psi},
 $$
 we obtain: $J_0(\rho)=\frac{1}{2\pi} \int_0^{2\pi} e^{i \rho \sin \psi}\, d\psi$, therefore it holds:
 $$
 V(\rho,z_0,z'_0)=\frac{e\,e'}{2\pi \, \e_f} \int_0^{+\infty} \frac{\cosh\left[k(\frac{d}{2}+z'_0) +\eta_1\right]\,\cosh[k(\frac{d}{2}-z_0)+\eta_2]}{\sinh\left(k\,d+\eta_1+\eta_2\right)} J_0(k\rho) dk ,
$$
 Finally, we change to the new variable: $u=k/d$ to get:

  \begin{equation}
 V(\rho,z_0,z'_0)=\frac{e\,e'}{2\pi \, \e_f d} \int_0^{+\infty} \frac{\cosh\left[u(\frac{1}{2}+\frac{z'_0}{d}) +\eta_1\right]\,\cosh[u(\frac{1}{2}-\frac{z_0}{d})+\eta_2]}{\sinh\left(u+\eta_1+\eta_2\right)} J_0(\frac{\rho}{d} u) du .
 \end{equation}
% \begin{figure}[ht] 
%\begin{center} 
%\includegraphics[width=12cm]{potentiel.eps}
%\caption{Surface plot of the interaction potentiel. Particle 1 is a position $(\rho=0,z=-.4)$.}\label{pot}
%\end{center} 
%\end{figure}

\end{document}